\documentclass[conference]{IEEEtran}
\IEEEoverridecommandlockouts
\usepackage{amsmath,amsfonts}

\usepackage{algorithm}
\usepackage{algorithmic}
\usepackage{amssymb}
\usepackage{multirow}
\usepackage{listings}%
\usepackage{tabularx}
\usepackage{textcomp}
\usepackage{stfloats}
\usepackage{url}
\usepackage{verbatim}
\usepackage{amsthm}
\usepackage{graphicx}
\usepackage{bbm}
\usepackage{cite}
\usepackage{color}
\usepackage[bookmarks=false]{hyperref}
\usepackage[letterpaper, top=0.7in, bottom=1in, left=0.625in, right=0.625in]{geometry}
\UseRawInputEncoding  
\usepackage{subcaption}
\usepackage{booktabs}
\usepackage{caption}
\usepackage{subcaption}
\newtheorem{assumption}{Assumption}
\newtheorem{theorem}{Theorem}

\newtheorem{corollary}{Corollary}

\renewcommand{\algorithmicrequire}{\textbf{Input:}}

\usepackage[T1]{fontenc}
\AtBeginDocument{
  \setlength{\abovedisplayskip}{3pt}
  \setlength{\belowdisplayskip}{3pt}
  \setlength{\abovedisplayshortskip}{2pt}
  \setlength{\belowdisplayshortskip}{2pt}
}
\begin{document}

\title{Consensus-based Decentralized Distributed Swarm Learning with Heterogeneous Big Data
\thanks{This work was supported in part by the National Science Foundation grants \#2146497, \#2231209, \#2244219, \#2315596, \#2343619, \#2349878, \#2416872, \#2413622, \#2533587, \#2548961, \#2551418, and \#2551417.} 
}
\author{\IEEEauthorblockN{Zhuoyu Yao$^\dagger$, \; Dong Yang$^\dagger$, \; Yue Wang$^\dagger$, \; Songyang Zhang$^*$, \; Yingshu Li$^\dagger$,    \; Zhi Tian$^\star$, \; Zhipeng Cai$^\dagger$}

\IEEEauthorblockA{$^\dagger$Department of Computer Science, Georgia State University, Atlanta, GA, USA\\$^*$Department of Electrical and Computer Engineering, University of Louisiana at Lafayette, LA, USA\\
$^\star$Department of Electrical and Computer Engineering, George Mason University, Fairfax, VA, USA}}

\maketitle

\begin{abstract}
Artificial intelligence increasingly relies on large-scale, distributed, and heterogeneous data collected by edge devices. 
However, the practice of edge intelligence remains challenging due to non-convex objectives, data heterogeneity, and complex wireless network topology. 
To address these issues, this paper proposes 
a consensus-based decentralized distributed swarm learning (CD-DSL) framework for wireless edge networks. 
Our CD-DSL integrates consensus optimization with particle swarm optimization (PSO), 
by reaching  
the 
model consensus among neighboring devices while leveraging 
the PSO exploration and exploitation. 
The consensus mechanism supports decentralized coordination without raw-data exchange, while PSO-inspired updates utilize historical and neighbor-shared experience to enhance exploration 
for non-convex optimization, improve robustness to data heterogeneity, and accelerate convergence.
We further develop an adaptive neighbor-mixing strategy that learns performance-aware consensus weights, improving decentralized collaboration among heterogeneous edge devices. 
Theoretical analysis establishes that CD-DSL maintains participant consistency and achieves non-ergodic convergence to a neighborhood of a stationary point under non-convex objectives.  
Experimental results show that CD-DSL can mitigate the performance degeneration of existing decentralized baselines  
caused by heterogeneous data.

\end{abstract}
\begin{IEEEkeywords}
Heterogeneous big data, consensus optimization, decentralized edge network, distributed swarm learning, non-convex optimization, convergence analysis.
\end{IEEEkeywords}

\section{Introduction}
The rapid growth of edge intelligence in wireless networks is shifting learning and inference from centralized platforms to large-scale edge devices, enabling low-latency and real-time intelligent services. 
However, big data generated by drones, sensors, vehicles, and mobile terminals  are inherently large-scale and highly heterogeneous, rendering centralized processing  impractical due to the prohibitive overhead for  data transmission and computation. 
Distributed learning avoids raw data collection, but 
experiences performance degradation under non-IID data and non-convex objectives~\cite{zhou2025distributed}.


Federated learning (FL) has been widely investigated to address data heterogeneity in distributed settings~\cite{wang2022accelerating}.  
Unfortunately, most pioneering FL frameworks hinge on a central server for model aggregation and coordination\cite{fan2021joint}, which is vulnerable to node/link failures and 
faces communication bottlenecks as well as concerns regarding trust in the central server\cite{yang2025physics}.
These limitations  motivate a more robust and efficient decentralization for edge intelligence\cite{xu2021coke}. However, replacing the central server with peer-to-peer collaboration does not automatically produce effective decentralized learning. 
Without a global coordinator, decentralized learning suffers from slower convergence and weaker agreement guarantees, especially with  large-scale devices and heterogeneous data.


In 
edge networks, heterogeneous device capabilities, asymmetric communication links, and dynamic wireless conditions 
lead to directed and imbalanced information flows~\cite{wu2025effectiveness}. 
Thus, conventional consensus-based learning methods may suffer from biased aggregation, degraded convergence, and limited robustness under realistic wireless network conditions.
While  decentralized learning under non-convex objectives has been studied in~\cite{yuan2016convergence,  kong2021consensus}, theoretical analysis remains underexplored. Even worse, data and system heterogeneity are intrinsic characteristics of wireless edge networks, where heterogeneity becomes more challenging in decentralized settings because each device only has a limited local perspective and cannot rely on a central server to provide a common evaluation criterion. 
Existing decentralized stochastic gradient descent (DSGD) methods primarily focus on improving local update rules~\cite{patel2024limits}, but 
not aiming at promoting 
collaboration among heterogeneous participants.
As a result, when local data distributions are highly non-IID, devices may produce conflicting model updates, causing slow convergence, poor generalization, and vulnerability to local traps in non-convex cases. 


Our recent 
work 
of 
distributed swarm learning (DSL) suggests an alternative 
direction to improve efficiency, robustness, and convergence~\cite{wang2024distributed, fan2023cb, yao2025dslota}.
By integrating AI with biologically inspired swarm intelligence, DSL leverages local historical experience and collaborators' information to improve learning efficiency under heterogeneous data distributions~\cite{wang2024distributed, yao2025multi}.
However, DSL is not originally developed for 
decentralized wireless networks as the existing designs often depend on a shared evaluation perspective to score particles or identify the best-performing model, which is not effectively 
available in serverless decentralized systems. 
Therefore, standard DSL lacks a principled aggregation mechanism for weighting neighbors over multi-hop communication graphs.

To address the above challenges, this paper proposes a consensus-based decentralized distributed swarm learning (CD-DSL) framework for robust and efficient edge intelligence systems with heterogeneous big data. 
CD-DSL extends distributed learning to fully decentralized edge systems by integrating consensus optimization with swarm learning, achieving faster convergence, efficient exploration and enhanced exploitation. We further develop an adaptive neighbor-mixing strategy to improve exploitation and collaboration under heterogeneous input data. Meanwhile, the swarm-intelligence component enables exploration to escape local optima for non-convex scenarios. Our contributions are listed as follows.

\vspace{-0.3mm}
\begin{enumerate}
\item  
In CD-DSL, we design a bio-inspired decentralized learning paradigm, which addresses non-IID big data 
and non-convex issues by integrating consensus-based decentralized learning with DSL exploration-and-exploitation without global evaluation.

\item We develop an adaptive neighbor-mixing strategy to learn adaptive consensus weights from local validation feedback. This mechanism enables each device to evaluate neighbors' models locally and assign higher weights to more beneficial collaborators, improving decentralized learning with heterogeneous data. 

\item We incorporate a push-sum consensus mechanism to correct aggregation bias caused by stochastic mixing over directed  graphs. 
This allows CD-DSL to operate in realistic asymmetric edge networks instead of relying on impractical doubly stochastic assumptions.

\item Theoretical analysis indicates that CD-DSL encourages participant consensus and achieves non-ergodic convergence to a neighborhood of a stationary point under non-convex objectives, and characterizes how data heterogeneity and hyperparameters affect convergence.

\end{enumerate}

\noindent \emph{Notations:} Let $\mathbf{1}\in \mathbb{R}^{n\times 1}$ be an all-one column vector and $\mathbf{I}\in \mathbb{R}^{n\times n}$ represent an identity matrix. We apply an event indicator $\mathbb{I}_{\left( {f_{i,t-1} > f_{i,t}} \right)}$, whose expectation over time $t$ describes the probability of ${f_{i,t-1} > f_{i,t}}$. 
Given a matrix  $\mathbf{A}\in \mathbb{R}^{n\times n}$, $\|\mathbf{A}\|$ is the spectral norm. And $\|\mathbf{x}\|$ denotes the Euclidean norm of a vector $\mathbf{x}\in \mathbb{R}^{n\times 1}$. 
Hadamard product and division of two matrices are expressed as $\mathbf{A} \circ \mathbf{B}$ and $\mathbf{A} \oslash \mathbf{B}$.

\section{CD-DSL: Consensus-based Decentralized Distributed Swarm Learning}\label{sec:Model}

This section formulates the decentralized learning problem, followed by a consensus mechanism design for directed topologies~\cite{yuan2016convergence}, which leads to our CD-DSL algorithm design. 

\vspace{-0.05in}
\subsection{Problem Formulation}
\label{subsec:problem}

Consider a decentralized wireless network with $K$  devices collaboratively performing a shared task. 
Its 
 topology can be represented as a strongly connected graph $\mathcal{G}=\left(\mathcal{V},\mathcal{E}\right)$ with the set of devices $\mathcal{V}$ and that of wireless communication links $\mathcal{E}$. 
The adjacency matrix $\mathbf{A}=\{a_{i,j}\}_{K\times K}$ characterizes the peer-to-peer communication accessibility, in which  $a_{i,j}=1$ denotes a reliable communication link from device $i$ to device $j$. Devices can only reach out their one-hop neighbors  to share model-training local updates 
(raw data query is prohibited). We set $a_{i,i}=1, \forall i$ for self-accumulated training gains from local historical experiences. Each participant maintains a local learning model $M_i(\mathbf{w}_i,D_i),\mathbf{w}_i \in \mathbb{R}^{N \times 1},i=1,\ldots,K$ with parameter vector $\mathbf{w}_i$ and local observation dataset $D_i$. 
It is worth noting that: caused by the biased local observations and heterogeneous data distributions, standalone training suffers from suboptimal global solutions, due to the lack of collaboration among decentralized devices. 
It motivates establishment of a global objective of  edge-intelligence systems, aiming a general model $\overline{\mathbf{w}}$ to mitigate heterogeneity impact: 
\begin{flalign} 
    &&  \overline{\mathbf{w}}& \triangleq  \arg\min_{\mathbf{w}}{\frac{1}{K}\sum\limits_{i=1}^{K}f_i(\mathbf{w},D_i)}, & \label{eq:objective}  
\end{flalign}
where loss function $f_i$ is evaluated based on local dataset $D_i$.  Local model updates follow a weighted local decentralized gradient descent mechanism for the generalization of $\overline{\mathbf{w}}$:
\begin{flalign}
    &&  {\mathbf{w}_{i,t+1}} =& \sum\limits_{k \in \mathcal{N}_i \cup \{i\}} {p_{k,i}\mathbf{w}_{k,t} } - \alpha_i \nabla f(\mathbf{w}_{i,t},D_i),\forall i, & \label{eq:concensus_updates}  
\end{flalign}
where $\alpha_i$ denotes the learning rate of $i$-th model, $p_{k,i}$ is a trainable parameter representing the alignment of neighboring devices $k$ to the device $i$. 
$ \mathcal{N}_i \cup \{i\}$ represents the neighborhood of device $i$, where $\mathcal{N}_i=\mathcal{N}_i^{in}\cup\mathcal{N}_i^{out}$ includes both in-degree and out-degree neighbors.  
For directed networks, $\mathcal{N}_i^{in}$ denotes the set of neighbors  from whom device $i$ can reliably receive messages, whereas $\mathcal{N}_i^{out}$ includes the neighbors to whom device $i$ can transmit its local updates. 
In this sense, each client can share its own historical experience and benefit from the neighbor's current training updates. 
A state space model is defined for the global iteration of Markov training process $\mathbf{W}_{t+1}=\mathbf{W}_{t}\cdot \mathbf{P}_{K \times K}-\Lambda , i=1,\ldots,K$, where $\mathbf{P}=[\mathbf{P}_i]=\{p_{i,j}\}$ is the mixing matrix to aggregate updates from neighbors.  $\mathbf{P}$ is initialized by the column normalizing $\mathbf{A}$ with $p_{i,j}^{t}\equiv 0,\forall t$ for all unconnected node pairs. Thus, $\sum\limits_{k \in \mathcal{N}_i \cup \{i\}} {p_{k,i}}=\sum_{k =1}^{k=K} {p_{k,i}},\forall i$ and we enlarge the integration interval from devices' neighborhood to the complete set for concise expression. $\Lambda\triangleq\left(\alpha_i \nabla f(\mathbf{w}_{i,t},D_i)\right)_1^{K}\in \mathbb{R}^{N\times K}$ denotes a stacked matrix with the local gradient, where $f(\mathbf{w}_{i,t},D_i)$ can be denoted as $f_{i,t}$. 

\subsection{Push Sum Consensus}
\label{subsec:Push_sum}

Consensus optimization orchestrates coordination for consistency of distributed local training through  consensus-based constraints to encourage  cooperation 
across participants.
Given \eqref{eq:concensus_updates}, the  problem  \eqref{eq:objective} in decentralized networks becomes: 
\begin{flalign}
    &&  \arg\min_{\mathbf{W}}\frac{1}{K} \mathbf{1}^{\top}\cdot\mathbf{f}(\mathbf{W},\mathbf{D})  &\triangleq{\frac{1}{K}\sum\limits_{i=1}^{K}f_i(\mathbf{w}_i,D_i)} & \nonumber \\
    &&  s.t. \quad \mathbf{w}_i &=\mathbf{w}_j, \forall \left(i,j\right)\in \mathcal{E}, & \label{eq:consensus_formula}  
\end{flalign}
where the vector $\mathbf{f}(\mathbf{W},\mathbf{D})\in \mathbb{R}^{1\times K}$ records the agents' loss functions and gradients of model parameters $\mathbf{W} \in \mathbb{R}^{N\times K}$.

The mixing matrix $\mathbf{P}$ performs 
the consensus-based aggregation scheme. 
To guarantee convergence, an assumption is usually applied  to hold a doubly stochastic $\mathbf{P}$ in order to aggregate updates from neighbors within undirected wireless topologies~\cite{nedic2020distributed}. 
However, the data and devices' heterogeneity in edge networks leads to asymmetric interactions, violating such an assumption of the symmetric and doubly stochastic property of $\mathbf{P}$. 
To fill this gap, we next investigate a robust consensus scheme for directed topologies, which is reliable to apply for heterogeneous data in wireless edge networks.  
\begin{assumption}
\label{ass_1}
(Column stochastic mixing matrix): 
Given a  
strongly connected network topology $\mathcal{G}$, there exists an optimal mixing matrix $\mathbf{P}^{\star}=\{p_{i,j}^{\star}\}\in \mathbb{R}^{K\times K}$ is nonnegative, off-diagonal entries and column stochastic as: 
\begin{flalign}
    &&  \mathbf{1}^{\top}\mathbf{P}&=\mathbf{1}; \quad 0< p_{i,j}, \forall (i,j) \in \mathcal{E}; \quad  0\leq p_{i,j} \leq1. & \label{eq:mixing_matrix}  
\end{flalign}
where eigenvalues  follow $1=|\lambda_1\left(\mathbf{P}\right)|\geq\ldots\geq|\lambda_K\left(\mathbf{P}\right)|$. 
\end{assumption}
Apart from the existing doubly stochastic mixing matrix design~\cite{yuan2016convergence, aketi2023global}, $\mathbf{P}$ forms a column stochastic, nonnegative and off-diagonal paradigm. 
The relaxed condition of {\bf Assumption~\ref{ass_1}} than that of \cite{kong2021consensus}, not only preserves its validity, but also makes it more suitable for directed and heterogeneous wireless edge networks. 
Consensus optimization enables finding a universal model to mimic the global optimal model such that all trained models will  converge to the same neighborhood of given stationary point $\mathbf{w}_1,\ldots,\mathbf{w}_K,\overline{\mathbf{w}} \in  \left[\mathbf{w}^{\star}-\epsilon, \mathbf{w}^{\star}+\epsilon\right]$. 

Although the column stochastic aggregation graph reduces modeling complexity and executes better in directed topological structure, the substituted row stochastic incurs an imperative gap as $1/K\sum\limits_{i}\mathbf{w}_{i,\infty}\neq 1/K\sum\limits_{j}\mathbf{w}_{j,\infty}$ as $\sum_{i=1}^{i=K}p_{i,j}\neq 1$. This gap incurs a biased objective $\overline{\mathbf{w}}= \arg\min_{\mathbf{w}}{1/K\sum\limits_{i}\Pr_i f_i(\mathbf{w}_i,D_i)}$ from the optimal aggregation, where the coefficient $\Pr_i$ is attributed to the deviant stationary distribution of $\mathbf{P}$.  Inspired by the push sum consensus mechanism~\cite{yuan2016convergence}, we introduce a nonnegative scalar matrix $\mathbf{Z}=\{\mathbf{z}_i\}_1^K,\mathbf{z}\in \mathbb{R}^{N \times 1}$ to rebuild symmetric  matrix for global updating. With the initialized $\mathbf{z}_{i,0}=\mathbf{1},\forall i$, individual parameter updates follow a ratio consensus prototype:
\begin{flalign}
     && \textstyle \mathbf{W}_{t+1}&=\left(\mathbf{W}_{t}\cdot  \mathbf{P}\right) \oslash\left(\mathbf{Z}_{t}\cdot  \mathbf{P}\right) -\Lambda & \nonumber \\ 
    && \textstyle \mathbf{Z}_{t+1}&=\mathbf{Z}_{t}\cdot  \mathbf{P}=\mathbf{Z}_{0}\cdot  \mathbf{P}  & \nonumber \\
    && \textstyle \widetilde{\mathbf{W}}_{t+1}&=\left(\mathbf{W}_{t+1}\circ\mathbf{Z}_{t+1}\right)^{\circ -1},  & \label{eq:push_sum}
\end{flalign}
where   $\mathbf{Z}$ reveals the accumulated aggregated gains on the scalar matrix. $\widetilde{\mathbf{W}}$ is the ratio of parameter $\mathbf{w}_{j,t}$ and scalar $\mathbf{z}_{j,t}$ for unbiased aggregation. Devices transmit $\widetilde{\mathbf{W}}$ as an adjusted individual gain to their neighbors. The receiving devices sum up all the information with column stochastic 
$\sum_{j=1}^{j=K}p_{j,i}=1$ to promote consistency.  Based on Perron-Frobenius~\cite{tsianos2012push} validation $\mathop{\lim}\limits_{t\to \infty}\mathbf{W}_{t} \circ \mathbf{Z}_{t}^{\circ -1}=1/K\mathbf{1}\mathbf{1}^{\top}\mathbf{W}_{0}$,  consensus algorithm $\overline{\mathbf{w}}_{t+1}=1/K\sum_{i=1}^{i=K}\left({\widetilde{\mathbf{w}}_{i,t}-\alpha_i \nabla f(\widetilde{\mathbf{w}}_{i,t})}\right)$ will converge to a specific neighborhood with proper learning rates and sufficient iteration. 
The push-sum mechanism improves consensus at the deployment of decentralized learning frameworks in directed wireless topologies. 

Although such a mechanism guarantees the generalization and non-ergodic convergence of decentralized learning paradigm for edge computing, data heterogeneity and non-convex objectives in wireless networks 
still prevent consensus based decentralized learning algorithm from robust and efficient implementation. 

\subsection{CD-DSL Algorithm}
\label{subsec:DDSL}
\vspace{-0.08in}


The prevalence of non-IID observations, non-convex downstream tasks, and one-hop communication constraints calls for a robust decentralized learning framework in wireless edge networks. Network intelligence improves model consistency by exploiting experience from neighboring participants~\cite{wu2025effectiveness}. Distributed swarm learning (DSL) has demonstrated strong robustness to heterogeneity and enhanced exploration and exploitation in non-convex optimization, making it a natural complement to consensus-based decentralized learning. The learning scheme of vanilla DSL implements~\cite{ yao2025multi}: 
\begin{flalign}
    &&  {\mathbf{v}_{i,t+1}} \!=& \mathbf{w}_{i,t+1} - \mathbf{w}_{i,t} & \nonumber \\
    &&   \!=& {c_0}{\mathbf{v}_{i,t}} + {c_1}\left( {\mathbf{w}_{i,t}^l - {\mathbf{w}_{i,t}}} \right) + {c_2}\left({\mathbf{w}_t^{g} - {\mathbf{w}_{i,t}}} \right) & \nonumber \\
    && &- \alpha \nabla F\left( {{\mathbf{w}_{i,t}}, D_g} \right), & \label{eq:DSL_update}  
\end{flalign}
where $\mathbf{v}_{i,t}$ is the parameter velocity. The exploration term ${c_0}{\mathbf{v}_{i,t}}$ enables the local model to perform exploratory updates that may temporarily deviate from locally optimal directions, strengthening the ability to escape saddle points. $\mathbf{w}_{i,t}^{l}$ and $\mathbf{w}_{t}^{g}$ are the local and global best models, which exploit each device's historical best solution and valuable experience from neighboring participants, and $c_0$, $c_1$, and $c_2$ control inertia, local, and global exploitation.

Although DSL improves robustness to heterogeneous data, its global best guidance is limited by the one-hop communication horizon in peer-to-peer networks. Facing this issue, to apply DSL in a decentralized manner, we introduce a consensus oriented neighbor best model to record the outstanding training gains from one-hop neighbors. 
The flooding-based information dissemination mechanism allows individual devices in a decentralized network to progressively obtain a global view through iterative neighbor-to-neighbor information exchanges~\cite{lim2001flooding}. 
Thanks to the long-term accumulation of neighbor-shared experience, the neighbor-best model gradually serves as an effective approximation to the global-best guide, leading the local models toward the global optimum. 
Based on the weighted parameter matrix $\widetilde{\mathbf{W}}$, the exploitation term is formulated as 
\begin{flalign}
    && \widetilde{\mathbf{W}}_{t}^l =&\left[\mathbb{I}_{\left( {f_{i,t} > f_{i,t-1}} \right)}\right]_{i=1}^K\cdot\mathbf{1}\circ\widetilde{\mathbf{W}}_{t-1} &\nonumber\\
    && &+ \left[\mathbb{I}_{\left( {f_{i,t-1} > f_{i,t}}  \right)}\right]_{i=1}^K\cdot\mathbf{1}\circ\widetilde{\mathbf{W}}_{t} &\label{eq:local_best} \\
    &&  \widetilde{\mathbf{W}}_{t}^{n}= & \widetilde{\mathbf{W}}_{t}\cdot\mathbf{P}, & \label{eq:neighbor_best}   
\end{flalign}
where the event indicator $b_{i,t}\triangleq \mathbb{I}_{\left( {f_{i,t-1} > f_{i,t}} \right)}$ denotes the probability that recent model on device $i$ performs better than previous one. As the training process follows a Markov chain, $\widetilde{\mathbf{W}}_{t}^l$ records the historical best parameters. Then the consensus-based decentralized DSL performs model updates: 
\begin{flalign}
    && \widetilde{\mathbf{W}}_{t+1}=& \widetilde{\mathbf{W}}_{t}\cdot\mathbf{P} -\Lambda(\widetilde{\mathbf{W}}_{t}) +{c_0}{\mathbf{V}_{t}}&\label{eq:Global_update} \\
    &&   &  + {c_1}\mathbf{B}_{t}\circ\left( {\widetilde{\mathbf{W}}_{t-1} - {\widetilde{\mathbf{W}}_{t}}} \right) + {c_2}\widetilde{\mathbf{W}}_{t}\cdot\left({\mathbf{P} - \mathbf{I}} \right), &   \nonumber
\end{flalign}
where the probability vector $\mathbf{B}=\left[\mathbb{I}_{\left( {f_{i,t-1} > f_{i,t}}  \right)}\right]\cdot\mathbf{1}\in \mathbb{R}^{K\times K}$ represents a binary matrix to formulate accumulated training gains, ${\mathbf{V}_{t}=\left[\mathbf{v}_{i,t}\right]_{i=1}^K}$ records the exploration velocity of participants to escape from local trap. $\mathbf{I}$ represents the identity matrix. For each local update:
\begin{flalign}
    && {\widetilde{\mathbf{w}}_{i,t+1}} =& (1+c_2)\sum\limits_{j} {p_{j,i}\widetilde{\mathbf{w}}_{j,t}} - \alpha_i \nabla f_{i,t}+{c_0}{\mathbf{v}_{i,t}} &\label{eq:local_update}\\
    &&   &  + {c_1}\mathbb{I}_{\left( {f_{i,t-1} > f_{i,t}} \right)}\left( {\widetilde{\mathbf{w}}_{i,t-1} - {\widetilde{\mathbf{w}}_{i,t}}} \right) - {c_2}\widetilde{\mathbf{w}}_{i,t} ,\forall i. & \nonumber
\end{flalign}
By incorporating the neighbor best model $\widetilde{\mathbf{w}}_{t}^{n}$, the mixing matrix serves not only to enforce consensus among local updates, but also to realize adaptive neighborhood aggregation. Accordingly, we design a joint optimization formulation to reflect the wireless communication topology and adaptive model comparison under heterogeneous inputs. At the current communication round $t$, client $i$ learns an adaptive incoming consensus weight vector $\mathbf{p}_i^t=\{p_{j,i}^t\}_{j=1}^K$ by evaluating neighbor models on its local validation data $D_i^{\mathrm{val}}$. The trainable consensus weights are obtained by solving
\begin{flalign}
    &&  \mathbf{p}_i^t=\arg\min_{\mathbf{p}_i}
&\sum_{j\in \mathcal{N}_i\cup\{i\}}
\left(p_{j,i} f(\mathbf{w}_{j,t},D_i) +\tau
p_{j,i}\log p_{j,i}\right) & \nonumber \\
&&  \mathrm{s.t.}\quad &\sum_{j=1}^{K}p_{j,i}=1, & \nonumber \\
&&  &p_{j,i}=0,\ \forall j\notin \mathcal{N}_i\cup\{i\}, & \nonumber \\
    && &p_{j,i}\geq 0, & \label{eq:optimization}  
\end{flalign}
where $f(\mathbf{w}_{j,t},D_i)$ measures the evaluation loss of neighbor model $\mathbf{w}_j^t$ on local validation dataset. This cross validation design provides a common evaluation of distributed model for heterogeneity measurement. $\tau>0$ controls the smoothness of the learned consensus weights. The entropy regularization $\tau\sum_{j\in \mathcal{N}_i\cup\{i\}}
p_{j,i}\log p_{j,i}$ avoids unstable hard selection of a single neighbor and preserves collaborative consensus. Based on the standard first-order optimality conditions~\cite{patel2024limits}, proposed entropy-regularized linear objective is strictly convex and yields a Gibbs-form closed-form solution as
\begin{equation}
p_{j,i}=\frac{\exp\left(-\frac{f_i(\mathbf{w}_{j,t},D_i)}{\tau}
\right)}{\sum\limits_{k\in \mathcal{N}_i\cup\{i\}}\exp\left(-\frac{f_i(\mathbf{w}_{k,t},D_i)}{\tau}\right)},\quad j\in \mathcal{N}_i\cup\{i\}.
\label{eq:softmax-p}
\end{equation}
In this way, CD-DSL forms a decentralized learning framework suitable for heterogeneous input in non-convex cases.

\vspace{-0.12in}

\subsection{Algorithm Implementation}
Algorithm~\ref{alg:CD-DSL} presents the implementation of CD-DSL. 
To reduce communication and computation overhead in peer-to-peer wireless networks, CD-DSL uses an asynchronous update schedule with $T$ communication rounds and $E$ local epochs per round. 
In each round, device $i$ receives model parameters and auxiliary variables from $\mathcal{N}_i \cup \{i\}$, evaluates neighboring models locally, and updates the mixing matrix $\mathbf{P}$ via \eqref{eq:softmax-p}. 
During local training, each device maintains its local and neighbor best parameters using \eqref{eq:local_best} and \eqref{eq:neighbor_best}, updates the model via \eqref{eq:local_update}, and transmits the updated variables to its neighbors. This asynchronous mechanism reduces synchronization overhead while maintaining decentralized collaboration.

\begin{algorithm}[!htb]
	\caption{CD-DSL}
	\label{alg:CD-DSL}
	\begin{algorithmic}[1]
\renewcommand{\algorithmicrequire}{\textbf{Initialization:}}
		\REQUIRE ~~\\
		Initialize $\overline{\mathbf{w}}_t, \mathbf{w}_{i,t}^l,\mathbf{w}_{t}^{n},\mathbf{z}_t$, $t,e,i$ and $T,E$, given datasets $\left\{ {{D_i}} \right\}$;
\\
    \FOR {each communication round $t=1:T$}

        \STATE \!\!\!\textbf{ at edge devices:$i =1,\ldots, K$}
            \STATE \hspace{0.1in} receive $\{\mathbf{w}_{j,t},{\mathbf{z}_{i,t}}\}$ from neighbors in $\mathcal{N}_i \cup \{i\}$ and calculate $\widetilde{\mathbf{w}}_{j,t}$;
            \STATE \hspace{0.1in} evaluate neighbors' models $f(\mathbf{w}_{j,t};D_i)$ via local dataset;
            \STATE  \hspace{0.1in} update $\mathbf{P}$ via \eqref{eq:softmax-p} for consensus and performance based aggregation;
            \FOR {each epoch $e=1:E$}
                \STATE   \hspace{0.1in}  update local and neighbor best via \eqref{eq:local_best} and \eqref{eq:neighbor_best}; 
                \STATE   \hspace{0.1in}  update local model via \eqref{eq:local_update}; 
                \STATE   \hspace{0.1in} transmit $\left\{ {{\mathbf{w}_{i,t + 1}},{\mathbf{z}_{i,t + 1}}} \right\}$ to neighbors;
            \ENDFOR
    \ENDFOR
	\end{algorithmic}
\end{algorithm}
\section{Theoretical Analysis}
In this section, we analyze  convergence behavior of CD-DSL. Since the weighted parameters $\widetilde{\mathbf{w}}$ represents a linear combination of $\mathbf{w}$, the convergence of $\mathbf{w}$ can be derived by $\widetilde{\mathbf{w}}$ with bounded scalar $\mathbf{z}$, which  has been proven in~\cite{LiuMomentum2020}. As such, we discuss the convergence of  $\widetilde{\mathbf{w}}$. Due to the page limit, we provide only a proof sketch here.  
Details can be found in an extended version of this paper\footnote{The detailed proof is available at \url{https://github.com/zhuoyu-3/CD-DSL/blob/main/Consensus_based_Distributed_Swarm_Learning__8_page.pdf}.}.
Let $L$-smoothness, $\sigma^2$-dissimilarity, $\phi^2$-bounded gradient and bounded local update assumptions hold. \textbf{Ergodic convergence:} We define the common model as $\overline{\mathbf{w}}_t \triangleq \frac{1}{K}\sum_{i=1}^{K}\widetilde{\mathbf{w}}_{i,t}$, the global objective $F(\overline{\mathbf{w}}) \triangleq \frac{1}{K}\sum_{i=1}^{K} f_i(\mathbf{w})$. The common model updates follow $\overline{\mathbf{w}}_{t+1}-\overline{\mathbf{w}}_t=\overline{\boldsymbol{\xi}}_t-\alpha_t\overline{\mathbf{g}}_t,\overline{\mathbf{g}}_t \triangleq \frac{1}{K}\sum_{i=1}^{K}\nabla f_i(\widetilde{\mathbf{w}}_{i,t}),\overline{\mathbf{\xi}}_t \triangleq \frac{1}{K}\sum_{i=1}^{K}\mathbf{\xi}_{i,t}$, where $\boldsymbol\xi_{i,t}=c_0\mathbf v_{i,t}+c_1 b_{i,t}\bigl(\widetilde{\mathbf w}_{i,t-1}-\widetilde{\mathbf w}_{i,t}\bigr)+c_2\left(\sum_{j=1}^{K} p_{j,i}^{t}\widetilde{\mathbf w}_{j,t}-\widetilde{\mathbf w}_{i,t}\right)$. \textbf{Non-ergodic convergence:} We also define the disagreement $\mathbf{E}_t\triangleq\widetilde{\mathbf{W}}_t\left(\mathbf{I}-\frac{1}{K}\mathbf{1}\mathbf{1}^{\top}\right)$, consensus error $e_t^2\triangleq\frac{1}{K}\|\mathbf{E}_t\|_F^2=\frac{1}{K}\sum_{i=1}^{K}\|\widetilde{\mathbf{w}}_{i,t}-\overline{\mathbf{w}}_t\|^2$. For the local updates $\widetilde{\mathbf{W}}_{t+1}=\widetilde{\mathbf{W}}_t\mathbf{P}_t-\alpha_t \mathbf{G}_t+\boldsymbol{\Xi}_t,\mathbf{G}_t \triangleq\left[\nabla f_1(\widetilde{\mathbf{w}}_{1,t}),\ldots,\nabla f_K(\widetilde{\mathbf{w}}_{K,t})\right],
 \boldsymbol{\Xi}_t \triangleq\left[\boldsymbol{\xi}_{1,t},\ldots,\boldsymbol{\xi}_{K,t}\right]$ the ergodic convergence of  $\overline{\mathbf{w}}$ is guaranteed by \textbf{Theorem~\ref{the:ergodic_convergence}}.
\begin{theorem}[Ergodic convergence]
\label{the:ergodic_convergence}
Let $\alpha_t$ satisfy
\begin{equation}
0<\alpha_t\le \bar\alpha\triangleq\min\left\{\frac{1}{4L_{\max}},\frac{1-\lambda}{2\sqrt{2}L_{\max}}\right\},
\label{eq:thm1_stepsize}
\end{equation}
where $\lambda\triangleq \sup_t |\lambda_2(P_t)|<1$.
Then, for any $T\ge 1$,

\begin{flalign}
    && &\frac{\sum_{t=0}^{T-1}\alpha_t\,\mathbb E\|\nabla F(\overline{\mathbf{w}}_t)\|^2}{\sum_{t=0}^{T-1}\alpha_t}\le\frac{4\left( F(\overline{\mathbf{w}}_0)-F(\overline{\mathbf{w}}^{\star})\right)}{\sum_{t=0}^{T-1}\alpha_t}\nonumber\\
 && &+\frac{3\sum_{t=0}^{T-1}\,\mathbb E[e_t^2]/\alpha_t}{8\sum_{t=0}^{T-1}\alpha_t}+
3\sigma^2+\frac{6\sum_{t=0}^{T-1}\,\mathbb E[\beta_t^2]/\alpha_t}{\sum_{t=0}^{T-1}\alpha_t} & \label{eq:ergodic_bound}
\end{flalign}

\end{theorem}

\begin{proof}
We assume bounded swarm perturbation term $\|\boldsymbol\xi_{i,t}\|\le\left(|c_0|+|c_1|+2|c_2|\right)\overline{\nu}_t\triangleq\beta_t$  and bounded consensus error and heterogeneous bias $\left\|\overline{\mathbf{g}}_t - \nabla F(\overline{\mathbf{w}}_t)\right\|^2\le 2L_{\max}^2 e_t^2 + \sigma^2$\cite{yao2025multi}. Given the $L$-smoothness of $F$, the accumulated expectation of $\nabla F(\overline{\mathbf{w}}_t)$ is bounded. To ensure convergence, we design $\alpha_t\leq\frac{1}{4L_{max}}$ for a negative coefficient of $\|\nabla F(\overline{\mathbf{w}}_t)\|^2$.
\end{proof}

\begin{theorem}[Bounded consensus  bias]\label{thm_consensus}
Let the learning rate satisfy $0<\alpha_t \le \frac{1-\lambda}{2\sqrt{2}L_{\max}}$ and define the spectral gap based parameter $\rho_t \triangleq \lambda+\sqrt{2}L_{\max}\alpha_t < 1$. The consensus error is
\begin{flalign}
    && &e_{t+1}\le\left(\lambda+\sqrt{2}L_{\max}\alpha_t\right)e_t+\sqrt{2}\sigma\alpha_t+\beta_t
&\label{eq:consensus_recursion}\\
 && &e_t\le\left(\prod_{s=0}^{t-1}\rho_s\right)e_0+\sum_{s=0}^{t-1} \left(\prod_{\ell=s+1}^{t-1}\rho_\ell\right) \left(\sqrt{2}\sigma\alpha_s+\beta_s\right),
&\label{eq:consensus_bound}
\end{flalign}
where $\beta_t$ is the effective swarm-induced perturbation bound.
\end{theorem}
\begin{proof}
Multiplying both sides of local updates by
$\mathbf{I}-\frac{1}{K}\mathbf{1}\mathbf{1}^{\top},$
we obtain $\mathbf{E}_{t+1}=\widetilde{\mathbf{W}}_t\left(\mathbf{P}_t-\frac{1}{K}\mathbf{1}\mathbf{1}^{\top}\right)-\alpha_t\mathbf{G}_t \left(\mathbf{I}-\frac{1}{K}\mathbf{1}\mathbf{1}^{\top}\right)+\boldsymbol{\Xi}_t\left(\mathbf{I}-\frac{1}{K}\mathbf{1}\mathbf{1}^{\top}\right)$. Substituting the spectral property $\left\|
\mathbf{P}_t-\frac{1}{K}\mathbf{1}\mathbf{1}^{\top}\right\|\le \lambda$, we have $e_{t+1}\le\;\lambda e_t+\alpha_t\left\|\overline{\mathbf{g}}_t - \nabla F(\overline{\mathbf{w}}_t)\right\|+\|\boldsymbol\xi_{i,t}\|$. Given bounded consensus error and heterogeneous bias and bounded swarm perturbation term, \eqref{eq:consensus_recursion} is yielded. Moreover, the bounded  series $\{e_t\}$ will converge with $\rho_t<1$ and yield \eqref{eq:consensus_bound}.
\end{proof}
Then, given Theorem~\ref{the:ergodic_convergence} and \ref{thm_consensus}, we further investigate the non-ergodic convergence of CD-DSL as follows.
\begin{corollary}[Non-ergodic convergence]\label{cor_constant}
Let $\alpha_t = \alpha_0\sqrt{\frac{K}{T}},\beta_t \le \overline{\beta}\alpha_t,$~\cite{aketi2023global}
with $\alpha_0$ satisfying \eqref{eq:thm1_stepsize}. The participant in the  $K$-clients CD-DSL system achieves
\begin{align}
\frac{1}{T}\sum_{t=0}^{T-1}\mathbb{E}\|\nabla F(\widetilde{\mathbf{w}}_{i,t})\|^2=&\mathcal{O}\left(\frac{1}{\alpha_0\sqrt{KT}}\right)+\mathcal{O}(\sigma^2)\nonumber \\
&+\mathcal{O}\!\left(\frac{K\alpha_0^2(\sigma+\overline{\beta})^2}{T}\right).
\label{eq:nonergodic_rate}
\end{align}
and the consensus bias obeys
\begin{equation}
\limsup_{t\to\infty} e_t\le\frac{K\alpha_0^2(\sqrt{2}\sigma+\overline{\beta})} {T(1-\lambda-\sqrt{2}L_{\max}\alpha_0\sqrt{K/T})}.
\label{eq:constant_neighborhood}
\end{equation}
\begin{equation}
e_t =\mathcal{O}\left(\frac{1}{\alpha_0\sqrt{KT}}\right) + \mathcal{O}(\rho^t),
\qquad
\rho < 1.
\label{eq:bias_rate}
\end{equation}
\end{corollary}

\noindent\textbf{Remark.}
Theorem~\ref{the:ergodic_convergence} reveals the common model $\overline{\mathbf{w}}_t$ will converge to $\overline{\mathbf{w}}^{\star}$. Theorem~\ref{thm_consensus} and Corollary~\ref{cor_constant} show that local models $\{\mathbf{w}_i\}_{i=1}^{K}$ will converge to an $\mathcal{O}( \sum_{s=0}^{t-1} \left(\prod_{\ell=s+1}^{t-1}\rho_\ell\right) \left(\sqrt{2}\sigma\alpha_s+\beta_s\right))$-neighborhood of  $\overline{\mathbf{w}}^{\star}$.
\vspace{-0.1in}

\begin{figure*}[!t]
    \centering
    \subfloat[Ablation study under IID data\label{results:ablation_iid}]{
        \includegraphics[width=2in]{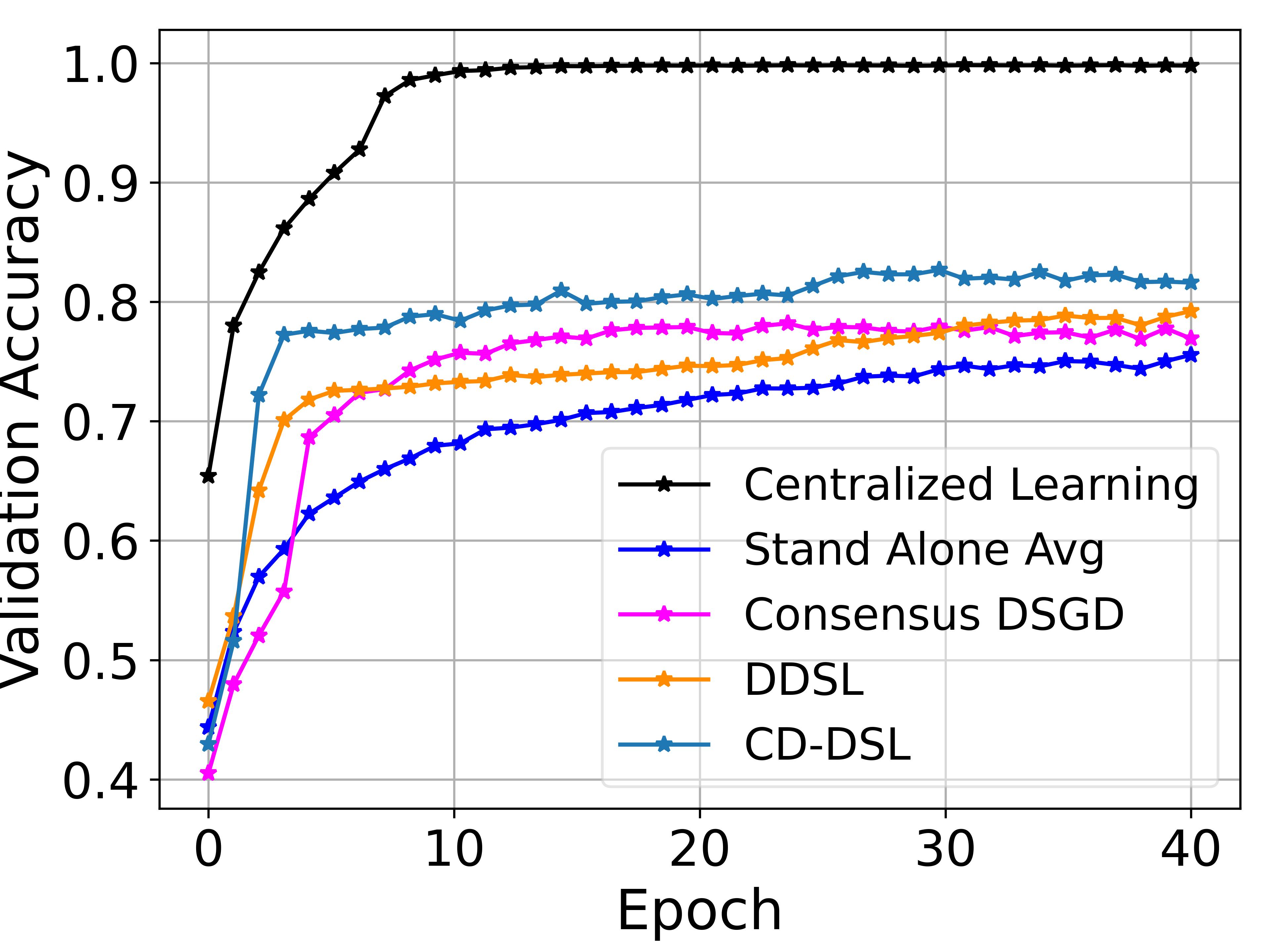}}%
    \hfil
    \subfloat[Ablation study under mild heterogeneity\label{results:ablation_mild}]{
        \includegraphics[width=2in]{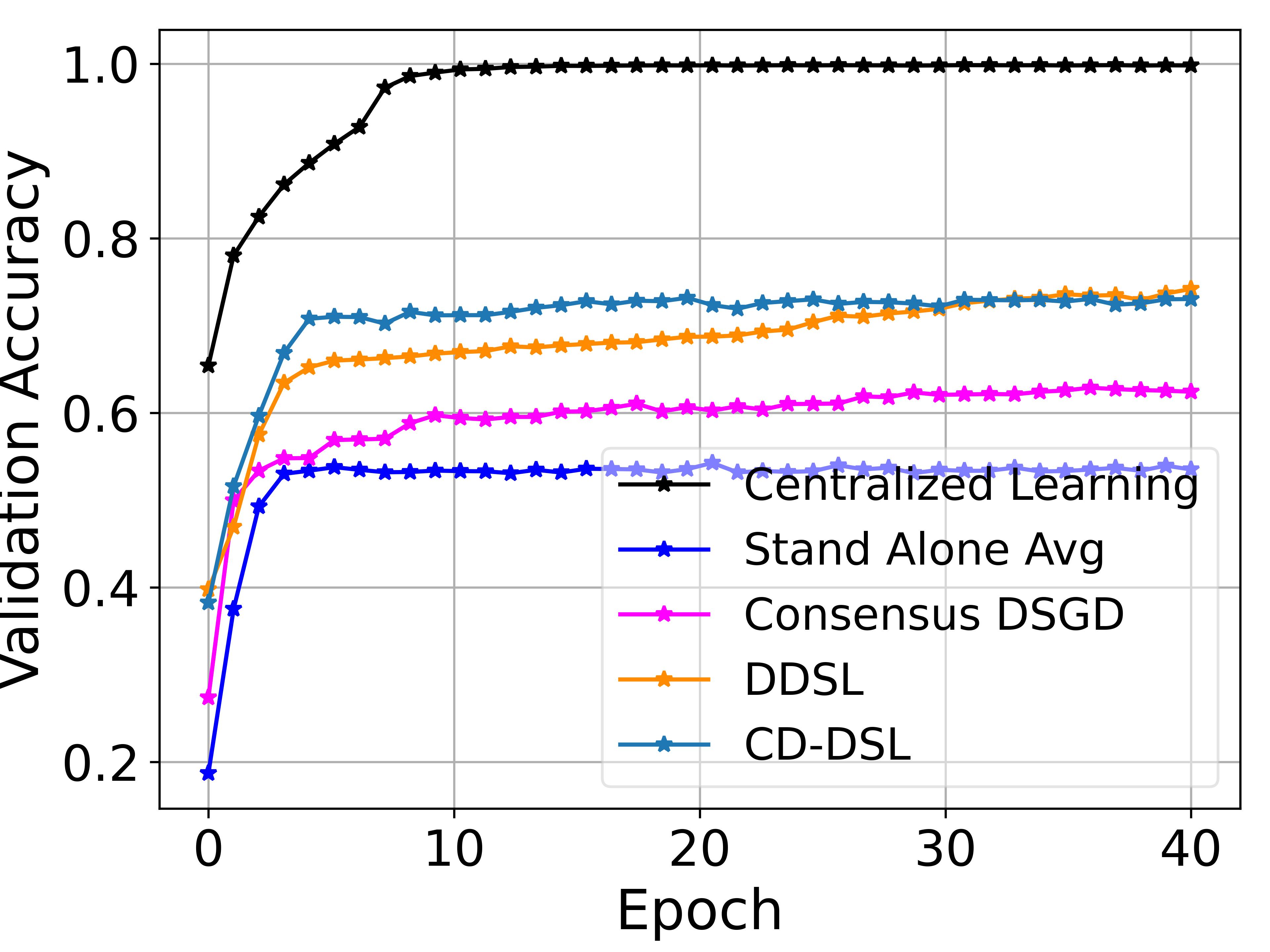}}%
    \hfil
    \subfloat[Ablation study under severe heterogeneity\label{results:ablation_severe}]{
        \includegraphics[width=2in]{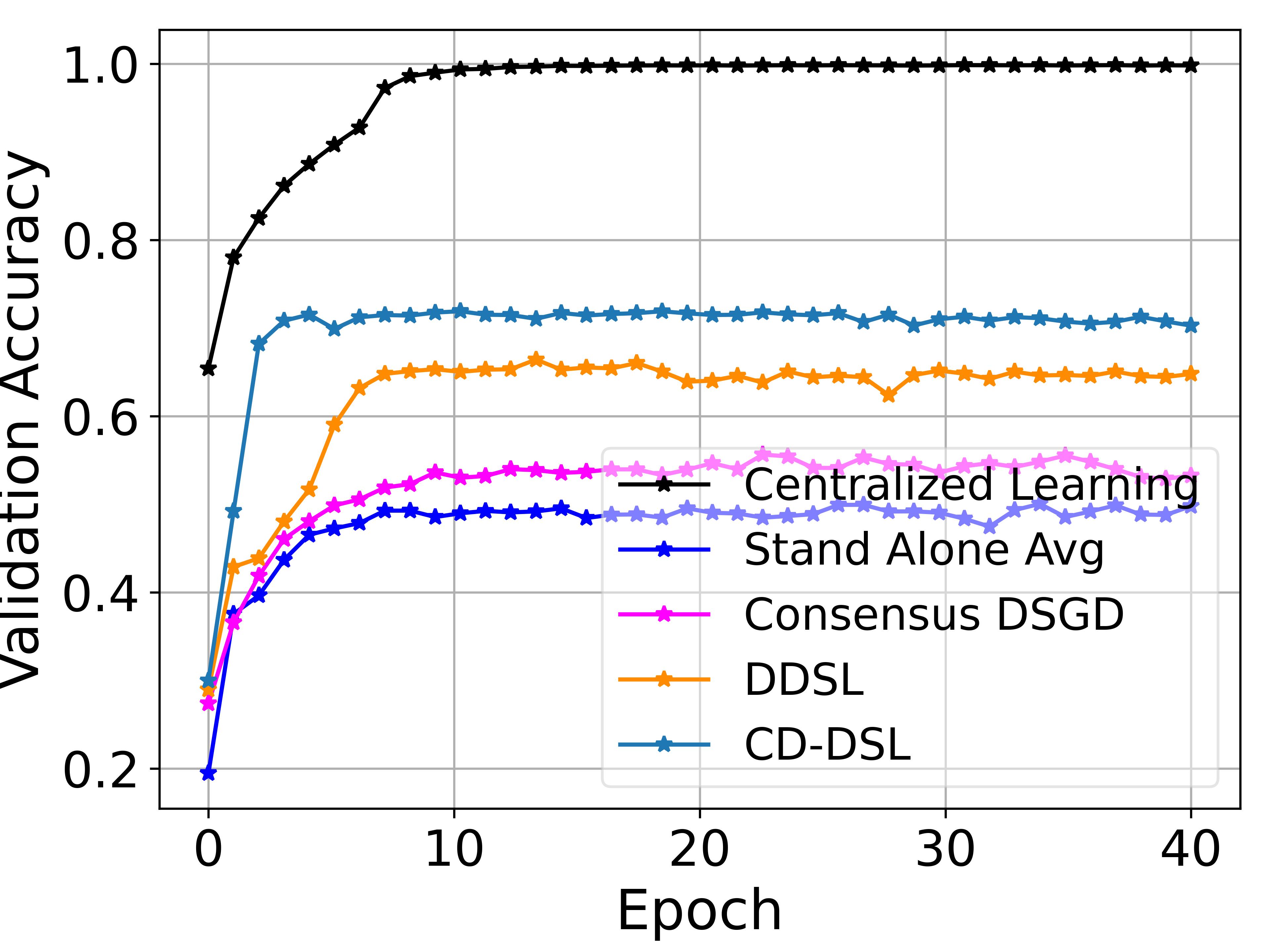}}\\
    \subfloat[Comparison under IID data\label{results:comp_iid}]{
        \includegraphics[width=2in]{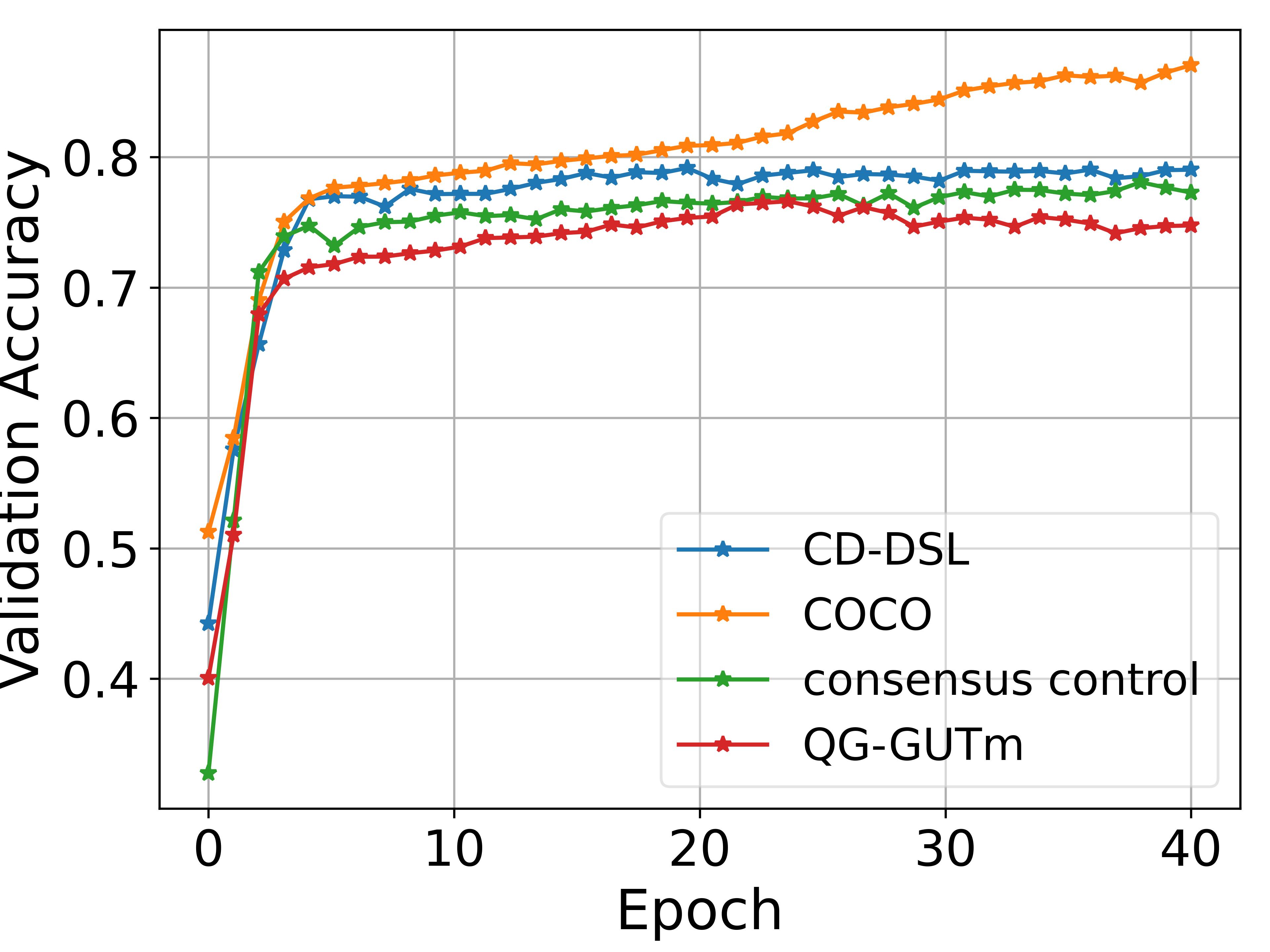}}%
    \hfil
    \subfloat[Comparison under mild heterogeneity\label{results:comp_mild}]{
        \includegraphics[width=2in]{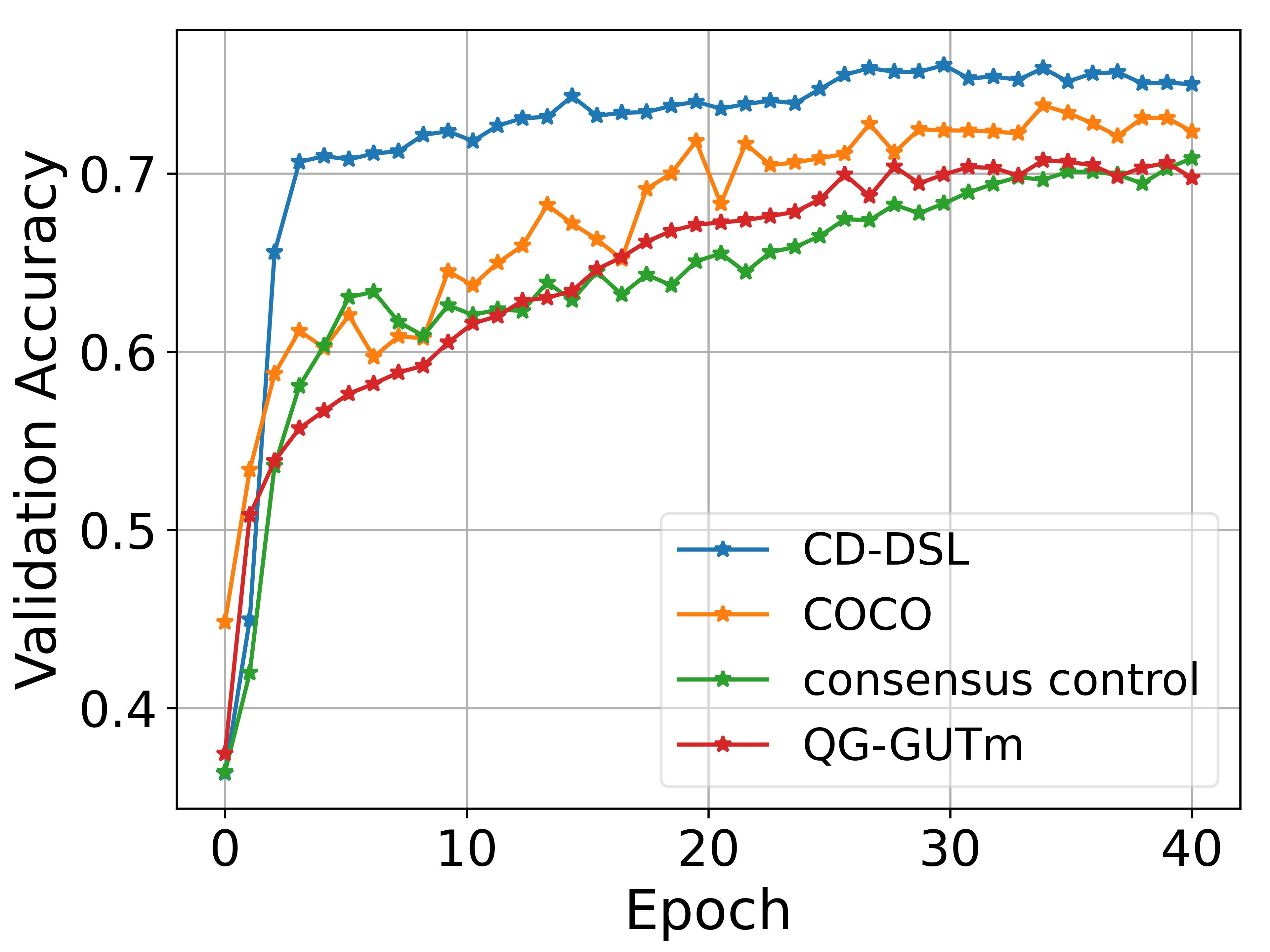}}%
    \hfil
    \subfloat[Comparison under severe heterogeneity\label{results:comp_severe}]{
        \includegraphics[width=2in]{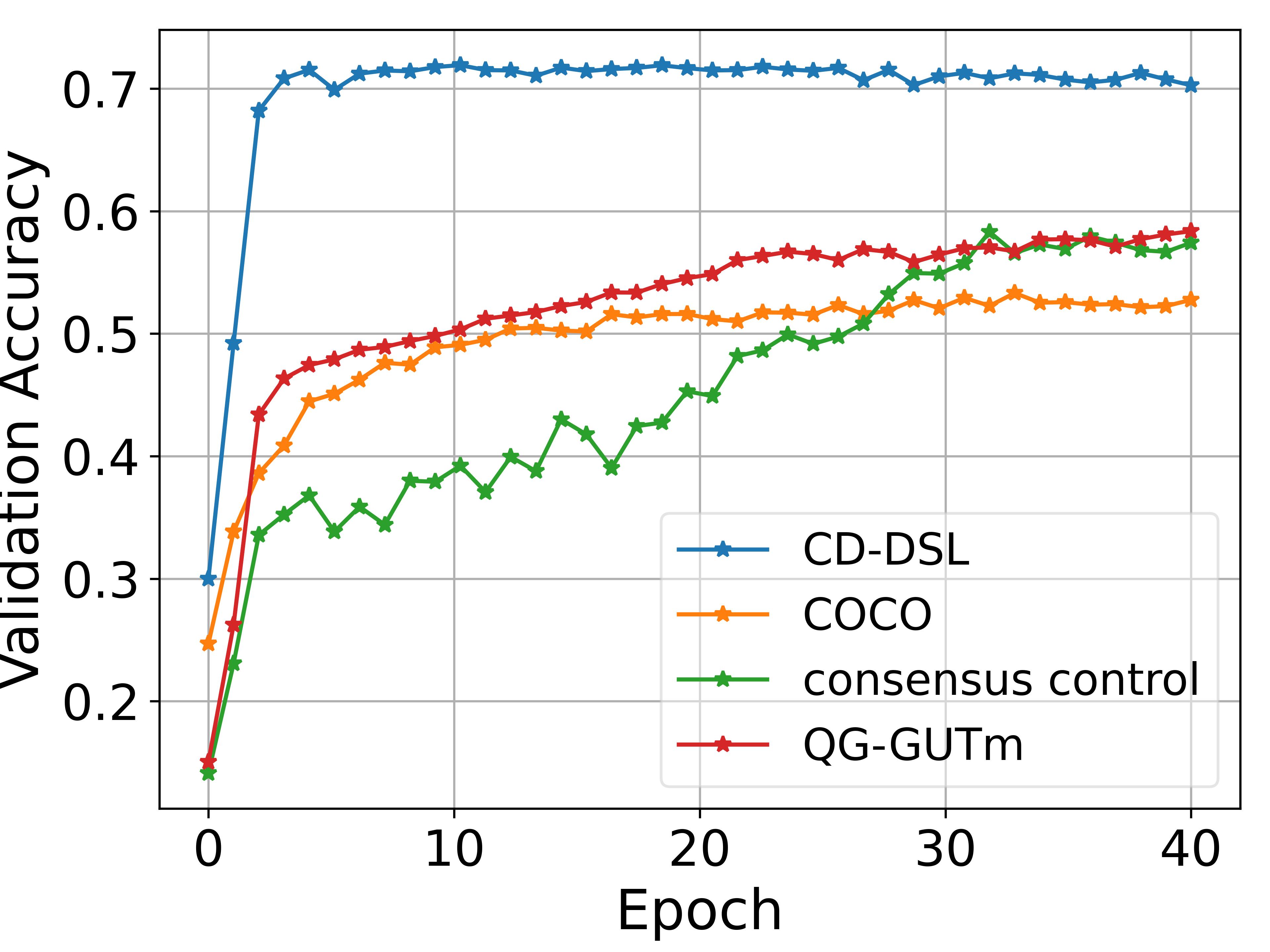}}%
    \vspace{-0.03in}
    \caption{Performance evaluation of the proposed CD-DSL compared with other benchmarks. 
    The first row test the ablation results of CD-DSL 
    under different data distributions. 
    The second row compares CD-DSL with other decentralized methods.}
    \label{results}
    \vspace{-0.15in}
\end{figure*}

\section{Simulation Results}
We test different decentralized learning methods via numerical experiments to train  ResNet18 for CIFAR10 image classification  with various settings. We use the cross-entropy loss to evaluate the non-convex optimization task.


\subsection{Experimental Settings}
We design a decentralized systems in Ubuntu  20.04.6 with 4 Nvidia V100-SXM2 GPUs to simulate the collaborative learning across $K=50$ edge devices\footnote{The code is available at {https://github.com/zhuoyu-3/CD-DSL}.}. Each contains an individual dataset ${\left| {{D_i}} \right|}=1000$ and a validation dataset ${\left| {{D_i^v}} \right|}=200$. Distributed data is deployed on a strongly connected decentralized edge network topology with a given connection rate of  $0.7$. The heterogeneous data generation relies on the Dirichlet distribution with concentration parameter $\alpha$. The sampled training and validation datasets share different label distributions across distributed devices. We design an IID case and two non-IID cases:  mild heterogeneity, where each tasks with $\alpha=0.5$ and severe heterogeneity, where 20 sets sampled by $\alpha=0.1$, 15 sets by $\alpha=0.5$, 10 sets by $\alpha=1$ and 5 sets by $\alpha=10$.

For the hyperparameter settings, we apply the SGD optimizer with attenuated learning rate $\alpha_{init}=0.01,\gamma=0.5$. The training process involves $40$ communication rounds on CIFAR10, each round includes $6$ epochs. The  batch size is $32$ for stable training. The hyperparameters of PSO updates follow $c_0 \sim U\left( {0,0.05(0.98)^{t}} \right)$ for controlled exploration and $c_1,c_2\sim \mathcal{N}\left( {0,0.15} \right)$ for balanced exploitation. The consensus weights in joint optimization  is set to $\tau=0.09$. Comparative experiments about COCO~\cite{wang2022accelerating}, consensus control~\cite{kong2021consensus}, and QG-GUTm~\cite{aketi2023global} are designed with optimal settings under given wireless topology. We evaluate and record the validation accuracy of the global model on an unseen validation dataset.

\subsection{Performance Evaluation and Comparison}

Fig.~\ref{results} reports the validation accuracy of CD-DSL compared with baseline 
methods under the IID case, mild heterogeneity, and the more challenging severe heterogeneity. 
The ablation results in Figs.~\ref{results:ablation_iid} to \ref{results:ablation_severe} show that the 
proposed CD-DSL consistently outperforms 
its ablated counterparts such as consensus-based DSGD\cite{yuan2016convergence} 
and decentralized distributed swarm learning (DDSL) as a straightforward but naive extension of CD-DSL arbitrarily applied in a decentralized manner\cite{wang2024distributed,fan2023cb}. 
Under the IID setting in Fig.~\ref{results:ablation_iid}, all decentralized variants can achieve stable convergence, while CD-DSL reaches higher accuracy with faster convergence than consensus-based DSGD. 
Such advantages become obvious as data heterogeneity increases. 
As shown in Fig.~\ref{results:ablation_mild}, CD-DSL maintains stable convergence and higher final accuracy, whereas DDSL saturates earlier. 
Such performance enhancement provided by CD-DSL beyond the benchmarks further increases under severe heterogeneity in Fig.~\ref{results:ablation_severe}, where CD-DSL continues to perform well but the ablated methods suffer from slower convergence and lower accuracy. 
These results indicate that consensus optimization improves model alignment, while DSL-enabled swarm updates enhance exploration and exploitation from non-IID data.

The comparison results in Figs.~\ref{results:comp_iid} to \ref{results:comp_severe} further validate the robustness of CD-DSL over the other decentralized baselines. 
Although COCO achieves competitive performance under the IID setting in Fig.~\ref{results:comp_iid}, CD-DSL 
is more robust against non-IID issues
as the data distribution becomes more heterogeneous. 
In the mild heterogeneity case, Fig.~\ref{results:comp_mild} shows that CD-DSL converges faster and achieves higher validation accuracy than COCO, consensus control, and QG-GUTm. 
This advantage becomes more evident in the severe heterogeneity case shown in Fig.~\ref{results:comp_severe}, where CD-DSL maintains a large performance margin over all baselines in terms of both convergence speed and converged accuracy. 
Overall, the results demonstrate that the integration of push-sum consensus, performance-aware mixing, and DSL-enabled swarm updates allows CD-DSL to achieve stronger model consistency and robustness in  decentralized learning with heterogeneous data.



\section{Conclusion}
This paper develops a consensus-based decentralized distributed swarm learning (CD-DSL) framework for edge intelligence with heterogeneous big data.
By integrating push-sum consensus, adaptive neighbor mixing, and bio-inspired model updates, the proposed CD-DSL enables decentralized learning over directed wireless topologies. Theoretical analysis shows that CD-DSL achieves  consensus and converges to a neighborhood of a stationary point under non-IID data distributions and non-convex objectives. Experimental results validate the robustness and efficiency of CD-DSL in heterogeneous decentralized learning scenarios.

\bibliographystyle{IEEEtran}
\bibliography{ref}

\end{document}